\documentclass[12pt]{article}
\usepackage{fullpage,amsmath, amssymb, amsfonts, amsthm, mathrsfs,mathtools,enumerate, enumitem,indentfirst}
\usepackage{bbm}
\usepackage{tikz}
\usetikzlibrary{arrows.meta}
\usepackage{xcolor}
\usepackage{braket}
\usepackage[ruled,lined]{algorithm2e}
\theoremstyle{definition}
\usepackage{hyperref}
\newtheorem{theorem}{Theorem}
\newtheorem{lemma}{Lemma}

\newtheorem{remark}[theorem]{Remark}

\newcommand{\KL}{\mathrm{KL}}

\title{Counterexample to a Proposed Capacity Characterization of the Relay Channel}
\author{Chun Hei Michael Shiu\thanks{Chun Hei Michael Shiu is with the Department of Electrical and Electronic Engineering, University of British Columbia, Vancouver, Canada. Email: {shiuchm@ece.ubc.ca}}}
\date{\today}
\makeatletter
\newcommand{\skipitems}[1]{%
  \addtocounter{\@enumctr}{#1}%
}
\begin{document}
\maketitle

\begin{abstract}
    We present a counterexample to the relay channel capacity characterization proposed in a recent work, which is based on properties of typical sequences. The counterexample uses a binary relay channel with a noiseless source-destination component and binary symmetric relay links. An achievable scheme exceeds the proposed capacity characterization, thereby disproving the characterization of the capacity for general relay channels. We highlight that the counterexample is obtained with the assistance of ChatGPT-6 Astra and the proofs are simplified with human effort.
\end{abstract}

\section{Introduction}

In the study of network information theory, a fundamental model of cooperative communication is the relay channel, introduced by van der Meulen in his seminal paper \cite{van71}. The relay channel models a point-to-point communication aided by an intermediate node, as arises in relay-assisted cellular networks for coverage extension \cite{kad10}. The characterization of the capacity for general relay channels has been one of the open problems in network information theory.

Upper and lower bounds on the capacity for relay channels have been extensively studied. The cutset upper bound is established in \cite{cov79}, which is known to be loose for certain relay channels. Several works have strengthened the cutset upper bound \cite{wu17improve, el22} or established tighter upper bounds \cite{wu17} for some specific classes of relay channels. On the other hand, lower bounds based on various schemes such as decode-forward, partial decode-forward, compress-forward are established. Matching bounds, and hence the capacity for some classes of relay channels are known. For instance, degraded \cite{cov79}, semi-deterministic \cite{gam82}, and relay channels with orthogonal sender components \cite{el05}.

In a recent work \cite{pon26}, the author claimed that the capacity of the relay channel has been solved (Theorem 4 in \cite{pon26}), which can be characterized by three relaying schemes: (i) decode-forward (DF), (ii) coordinated compress-forward (C-CF), and (iii) uncoordinated compress-forward (U-CF). The proposed optimality proof relies on typicality. Specifically, the converse (outer bound) derives a compression-cutset mutual information bound through typical sequence analysis and a sphere-packing argument; the achievability (inner bound) combines the decode-forward, uncoordinated compress-forward without Wyner-Ziv coding, and coordinated compress-forward scheme as encoding, with a joint typicality decoding.

However, we have identified a flaw in the proposed converse (outer bound) in Section V of \cite{pon26}: an unjustified identification of the distribution induced by an arbitrary block relay encoder with an auxiliary compression distribution. All equation numbers in the remainder of this paragraph refer to the equations in \cite{pon26}. In the proof, after defining the compression codebook and relay compression mapping as in Eq. (127) and Eq. (128), respectively, they assign the approximate probability $2^{-n_k H(\hat{Y}_r |X_r)}$ to the actual compression-codebook elements as in Eq. (133), invoking the typicality inclusion in Eq. (134). However, membership in a typical set does not determine the corresponding codeword's selection probability under the relay mapping, as the probability under the auxiliary product distribution need not coincide with its probability under the distribution induced by the block encoder. Consequently, the conditional probability and estimates in Eqs. (136) -- (138) are not justified, and the argument fails to establish the claimed reduction to single-letterized compression.

In this paper, we present an explicit counterexample to the capacity characterization proposed in \cite{pon26} for relay channels. Specifically, we exhibit a relay channel for which a partial decode-forward scheme achieves a rate strictly larger than the proposed expression, thereby disproving the claimed characterization and illustrating the consequence of the flaw. The remainder of this paper is organized as follows. In Section \ref{sec: relay channel} we introduce the relay-channel model and relevant preliminaries. In Section \ref{sec: counterexample and proofs} we present the counterexample and establish the strict gap between an achievable rate and the characterization in \cite{pon26}. Section \ref{sec: concluding remarks} concludes the paper with a discussion.

\subsection*{Notation}
In this paper, we use calligraphic letters to denote sets (alphabets), capital letters to denote random variables, and lowercase letters to denote realizations of the corresponding random variables. All logarithm and entropic quantities are to base $2$, unless stated otherwise.
We use $h_2(\cdot)$ to denote the binary entropy function. We denote $\mathrm{BSC}(\delta)$ to be the binary symmetric channel (BSC) with crossover probability $\delta \in [0,1]$.

\section{Relay Channel and Proposed Capacity} \label{sec: relay channel}
In this paper, we consider the three-node discrete memoryless relay channel (DM-RC). A sender wishes to transmit a message $M$ to the receiver with the help of a relay. The DM-RC is specified by four finite alphabets: $\mathcal X_s$ (source input alphabet), $\mathcal X_r$ (relay input alphabet), $\mathcal Y_r$ (relay output alphabet), $\mathcal Y_d$ (destination output alphabet),\footnote{Here input and output are viewed from the channel's perspective.} together with a collection of conditional probability mass functions (pmfs) $p(y_r, y_d|x_s, x_r)$ on $\mathcal Y_r \times \mathcal Y_d$ for each $(x_s, x_r) \in \mathcal X_s \times \mathcal X_r$. A relay channel model is depicted in Figure \ref{fig:relay-channel}.

\begin{figure}[t]
    \centering
    \begin{tikzpicture}[
        >=Stealth,
        line width=0.7pt,
        block/.style={
            draw,
            minimum height=1cm,
            align=center,
            inner sep=8pt
        }
    ]
        % Nodes
        \node[block, minimum width=2cm]
            (encoder) at (0,0) {Encoder};

        \node[block, minimum width=3cm]
            (channel) at (3.8,0)
            {$p(y_r,y_d\mid x_s,x_r)$};

        \node[block, minimum width=3cm]
            (relay) at (3.8,2.4) {Relay encoder};

        \node[block, minimum width=2cm]
            (decoder) at (7.6,0) {Decoder};

        % Source and destination links
        \draw[->] (-2.3,0)
            -- node[above] {$M$} (encoder.west);

        \draw[->] (encoder.east)
            -- node[above] {$X_s^n$} (channel.west);

        \draw[->] (channel.east)
            -- node[above] {$Y_d^n$} (decoder.west);

        \draw[->] (decoder.east)
            -- node[above] {$\hat{M}$} (9.9,0);

        % Relay observation and transmission links
        \draw[->] ([xshift=-0.9cm]channel.north)
            -- node[left] {$Y_r^n$}
            ([xshift=-0.9cm]relay.south);

        \draw[->] ([xshift=0.9cm]relay.south)
            -- node[right] {$X_r^n$}
            ([xshift=0.9cm]channel.north);
    \end{tikzpicture}
    \caption{Point-to-point Communication with a Relay}
    \label{fig:relay-channel}
\end{figure}
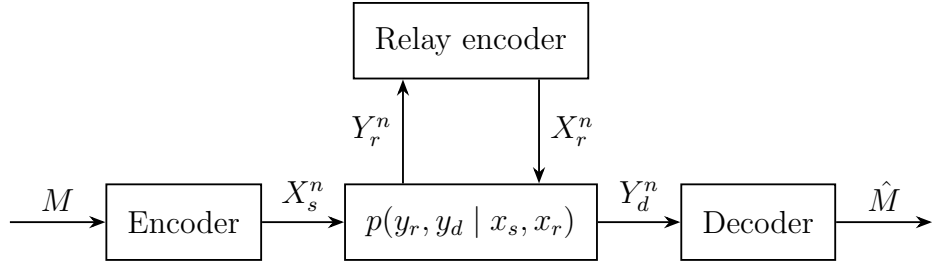

For a positive integer $n$ and a rate $R \geq 0$, an $(n, R)$ code for DM-RC consists of 
\begin{itemize}
    \item a message set $[1 : 2^{nR}] \triangleq \{1, 2, \ldots, 2^{nR}\}$;
    \item a source encoder that maps each message $m \in [1:2^{nR}]$ to a codeword $x_s^n(m)$;
    \item a sequence of relay encoders $f_{ri} : \mathcal Y_r^{i-1} \to \mathcal X_r$, $i \in [1:n]$, such that $X_{ri} = f_{ri}(Y_r^{i-1})$;
    \item a decoder that estimate $\hat{m}$ or output an error symbol $\mathrm{e}$ to the received sequence $y_d^n \in \mathcal Y_d^n$.
\end{itemize}

Assuming that $M$ is uniformly distributed over the message set, a rate $R$ is said to be achievable for a DM-RC if there exists a sequence of $(n, R)$ codes such that $\lim_{n \to \infty} \Pr(\hat{M} \neq M) = 0$. The capacity of a DM-RC is the supremum of all achievable rates.

In the following, we introduce the results presented in \cite{pon26}. Before stating the main result, we define the necessary quantities and expressions. In the paper, they considered an auxiliary random variable $\hat{\mathcal Y}_r$ at the relay for implementing relay compressions. Compress-forward (CF) schemes 
% the joint distribution on $\mathcal X_s \times \mathcal X_r \times \mathcal Y_r \times \hat{\mathcal Y}_r \times \mathcal Y_d$ can be factorized in the following way:
% \begin{align*}
%     p(x_s, x_r, y_r, \hat{y}_r, y_d) & = p(x_s, x_r) \cdot p(y_r, y_d | x_s, x_r) \cdot p(\hat{y}_r | y_r, x_r) \\
%     & =  p(x_s, x_r) \cdot p(y_r, y_d | x_s, x_r) \cdot p(\hat{y}_r | x_r) \cdot p(y_r | \hat{y}_r, x_r). 
% \end{align*}
% Therefore, the scheme 
can be characterized by a triple of distributions: (i) the input distribution $p_{\mathrm{cf}}(x_s, x_r)$ over $\mathcal X_s \times \mathcal X_r$, (ii) the compression prior $p_{\mathrm{cf}}(\hat{y}_r | x_r)$ over $\hat{\mathcal Y}_r \times \mathcal X_r$, and (iii) the reverse-compression channel $p_{\mathrm{cf}}(y_r |\hat{y}_r, x_r)$ over $\mathcal Y_r \times \hat{\mathcal Y}_r \times \mathcal X_r$. Define the set of CF feasible distribution triplets be
\begin{align*}
    \mathcal P_{\mathrm{cf}} := \{\mathbf{p} =( p_{\mathrm{cf}}(x_s,x_r), p_{\mathrm{cf}}(y_r | \hat{y}_r , x_r), p_{\mathrm{cf}}(\hat{y}_r|x_r))\}. %\label{eq: total feasible set}
\end{align*}
We also define the following forward relay channel:
\begin{align*}
    p_{\mathrm{for}}(y_r|\hat{y}_r, x_r) := \frac{\sum_{x_s, y_d} p_{\mathrm{cf}}(x_s, x_r) \cdot p(y_r, y_d | x_s, x_r) \cdot p_{\mathrm{cf}}(\hat{y}_r | x_r)}{\sum_{x_s, \hat{y}_r, y_d} p_{\mathrm{cf}}(x_s, x_r) \cdot p(y_r, y_d | x_s, x_r) \cdot p_{\mathrm{cf}}(\hat{y}_r | x_r)} %\label{eq: forward-compression channel}
\end{align*}

We say that the CF distribution triplets satisfy the independent condition if
\begin{align}
    p_{\mathrm{cf}}(x_s, x_r) =  p_{\mathrm{cf}}(x_s) \cdot p_{\mathrm{cf}}(x_r), \label{eq: independent condition}
\end{align}
and we say that they satisfy the correlated condition if
\begin{align}
    p_{\mathrm{cf}}(y_r|\hat{y}_r, x_r) = p_{\mathrm{for}}(y_r | \hat{y}_r, x_r). \label{eq: correlated condition}
\end{align}
With these conditions, we may define the set of independent feasible distribution triplets and the set of correlated feasible distribution triplets to be 
\begin{align*}
    \mathcal P_{\mathrm{ind}} &:= \{ \mathbf{p} = (p_{\mathrm{cf}}(x_s, x_r), p_{\mathrm{cf}}(y_r | \hat{y}_r, x_r), p_{\mathrm{cf}}(\hat{y}_r | x_r)) : \mathbf{p} \text{ satisfies the independent condition}\}, \\
    \mathcal P_{\mathrm{cor}} & := \{\mathbf{p} = (p_{\mathrm{cf}}(x_s, x_r), p_{\mathrm{cf}}(y_r | \hat{y}_r, x_r), p_{\mathrm{cf}}(\hat{y}_r | x_r)) : \mathbf{p} \text{ satisfies the correlated condition}\},
\end{align*}
respectively. Using the reverse-compression channels, we may also define the forward-compression channels as 
\begin{align*}
    p_{\mathrm{for}}(\hat{y}_r | y_r, x_r) & := \frac{\sum_{x_s, x_r}p_{\mathrm{cf}}(x_s, x_r) \cdot p_{\mathrm{cf}}(\hat{y}_r|x_r) \cdot p_{\mathrm{for}}(y_r | \hat{y}_r, x_r)}{\sum_{x_s, x_r, y_r}p_{\mathrm{cf}}(x_s, x_r) \cdot p_{\mathrm{cf}}(\hat{y}_r|x_r) \cdot p_{\mathrm{for}}(y_r | \hat{y}_r, x_r)}, \\
    p_{\mathrm{cf}}(\hat{y}_r | y_r, x_r) & := \frac{\sum_{x_s, x_r}p_{\mathrm{cf}}(x_s, x_r) \cdot p_{\mathrm{cf}}(\hat{y}_r|x_r) \cdot p_{\mathrm{cf}}(y_r | \hat{y}_r, x_r)}{\sum_{x_s, x_r, y_r}p_{\mathrm{cf}}(x_s, x_r) \cdot p_{\mathrm{cf}}(\hat{y}_r|x_r) \cdot p_{\mathrm{cf}}(y_r | \hat{y}_r, x_r)}.
\end{align*}
Hence, we have the following joint distributions on $\mathcal X_s \times \mathcal X_r \times \mathcal Y_r \times \hat{\mathcal Y}_r \times \mathcal Y_d$:
\begin{align}
    p_{\mathrm{for}}(x_s, x_r, y_r, \hat{y}_r, y_d) & := p_{\mathrm{cf}}(x_s, x_r) \cdot p(y_r, y_d| x_s, x_r) \cdot p_{\mathrm{for}}(\hat{y}_r | y_r ,x_r), \\ %\label{eq: joint for} \\
    p_{\mathrm{cf}}(x_s, x_r, y_r, \hat{y}_r, y_d) & := p_{\mathrm{cf}}(x_s, x_r) \cdot p(y_r, y_d| x_s, x_r) \cdot p_{\mathrm{cf}}(\hat{y}_r | y_r ,x_r). %\label{eq: joint cf}
\end{align}

In the following, we use $H$ and $I$ to denote the standard Shannon entropy and mutual information, while the subscript ``for'' and ``cf'' refers to the joint distribution that we are evaluating on. Define the uncoordinated-CF feasible set and coordinated-CF feasible set as
\begin{align*}
    \mathcal P_{\mathrm{ind-feasible}}:=\{& (p(x_s, x_r) , p(y_r | \hat{y}_r, x_r), p(\hat{y}_r|x_r)) \in \mathcal P_{\mathrm{ind}} : \\
    & H_{\mathrm{cf}}(Y_r | X_r) \geq H_{\mathrm{for}}(Y_r|X_r) \text{ and } I_{\mathrm{cf}}(\hat{Y}_r ; Y_r | X_r) < I_{\mathrm{cf}}(X_r; Y_d)\}, \\
    \mathcal P_{\mathrm{cor-feasible}}:=\{& (p(x_s, x_r) , p(y_r | \hat{y}_r, x_r), p(\hat{y}_r|x_r)) \in \mathcal P_{\mathrm{cor}} : \\
    & H_{\mathrm{cf}}(Y_r | X_r) \geq H_{\mathrm{for}}(Y_r|X_r) \text{ and } I_{\mathrm{for}}(\hat{Y}_r ; Y_r | X_r) < I_{\mathrm{for}}(\hat{Y}_r, X_r; Y_d)\}.
\end{align*}

Now we are ready to present the main theorem. 
\begin{theorem}[Theorem 4 in \cite{pon26}] \label{thm: main thm}
    For a relay channel $(\mathcal X_s,\mathcal X_r, \mathcal Y_r, \mathcal Y_d, \{p(y_r, y_d | x_s, x_r)\})$, the capacity of the relay channel is 
    \begin{align*}
        C_{\mathrm{prop}} = \max\{R_{\mathrm{U-CF}}, R_{\mathrm{C-CF}}, R_{\mathrm{DF}}\},
    \end{align*}
    where the uncoordinated-compress-forward rate $R_{\mathrm{U-CF}}$ is defined as
    \begin{align*}
        R_{\mathrm{U-CF}} := \max_{(p(x_s, x_r), p(y_r |\hat{y}_r, x_r), p(\hat{y}_r|x_r)) \in \mathcal P_{\mathrm{ind-feasible}}} I_{\mathrm{cf}}(X_s; \hat{Y}_r, Y_d | X_r);
    \end{align*}
    the coordinated-compress-forward rate $R_{\mathrm{C-CF}}$ is defined as 
    \begin{align*}
        R_{\mathrm{C-CF}} := \max_{(p(x_s, x_r), p(y_r |\hat{y}_r, x_r), p(\hat{y}_r|x_r)) \in \mathcal P_{\mathrm{cor-feasible}}} I_{\mathrm{for}}(X_s; \hat{Y}_r, Y_d | X_r);
    \end{align*}
    and the decode-forward rate $R_{\mathrm{DF}}$ is defined as 
    \begin{align*}
        R_{\mathrm{DF}} := \max_{p(x_s, x_r)} \min\{I(X_s; Y_r | X_r), I(X_s, X_r; Y_d)\}
    \end{align*}
    with joint distribution $p(x_s, x_r, y_r, y_d) = p(x_s, x_r) \cdot p(y_r, y_d | x_s, x_r)$.
\end{theorem}

\section{Counterexample and Proofs} \label{sec: counterexample and proofs}
In this section, we present a relay channel whose capacity is strictly larger than the proposed capacity in Theorem \ref{thm: main thm}, thereby serving as a counterexample to Theorem 4 in \cite{pon26}.% Again, we emphasize that the counterexample is constructed by ChatGPT-6 Astra.

\subsection{Counterexample Relay Channel} \label{sec: counterexample channel}
Consider the following relay channel $\mathcal R$. Let the (sender) source input be $X_s = (A,B) \in \{0,1\}^2$ be a binary pair, and also let the relay input $X_r \in \{0,1\}$ be binary. The relay channel outputs
\begin{align*}
    Y_r = A \oplus N_1, \qquad Y_d = (B, Z) = (B, X_r \oplus N_2),
\end{align*}
where $N_1, N_2 \sim \mathrm{Bern}(1/4)$ are mutually independent noises, independent of the current channel inputs, with fresh independent copies across channel uses. It is clear that $\mathcal X_s = \{0,1\}^2 = \mathcal Y_d$, $\mathcal X_r = \{0,1\} = \mathcal Y_r$. The relay channel $\mathcal R$ is depicted in Figure \ref{fig:counterexample-relay}.

\begin{figure}[ht]
    \centering
    \begin{tikzpicture}[
        >=Stealth,
        line width=0.7pt,
        block/.style={
            draw,
            minimum height=0.85cm,
            align=center,
            inner sep=6pt
        },
        lab/.style={
            midway,
            above=4pt,
            inner sep=1pt
        }
    ]
        % Main blocks
        \node[block, minimum width=1.8cm]
            (encoder) at (-5,0) {Encoder};

        \node[block, minimum width=3.8cm]
            (channel) at (0,0)
            {$p(y_r,(b,z)\mid(a,b),x_r)$};

        \node[block, minimum width=1.8cm]
            (decoder) at (5,0) {Decoder};

        \node[block, minimum width=2.8cm]
            (relay) at (0,2.8) {Relay encoder};

        % Horizontal links
        \draw[->] (-7,0)
            -- node[lab] {$M$} (encoder.west);

        \draw[->] (encoder.east)
            -- node[lab] {$(A,B)$} (channel.west);

        \draw[->] (channel.east)
            -- node[lab] {$(B,Z)$} (decoder.west);

        \draw[->] (decoder.east)
            -- node[lab] {$\hat M$} (7,0);

        % Vertical relay links
        \draw[->] ([xshift=-0.9cm]channel.north)
            -- node[midway,left=5pt] {$Y_r$}
            ([xshift=-0.9cm]relay.south);

        \draw[->] ([xshift=0.9cm]relay.south)
            -- node[midway,right=5pt] {$X_r$}
            ([xshift=0.9cm]channel.north);
    \end{tikzpicture}
    \caption{Counterexample relay channel $\mathcal R$.}
    \label{fig:counterexample-relay}
\end{figure}
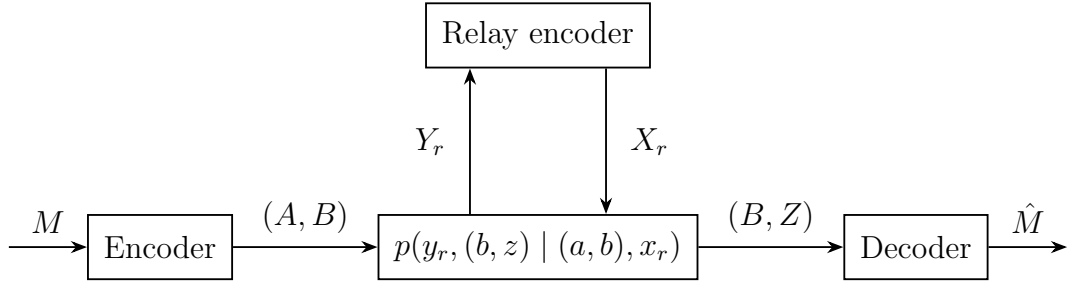

\subsection{Lower Bound for the Capacity of $\mathcal R$} \label{sec: lower bound}

In this section, we describe an achievability scheme for the above relay channel, hence illustrating a lower bound for its capacity. 

\begin{lemma}
    Let $\mathcal R$ be the relay channel defined in Section \ref{sec: counterexample channel}. Then $C(\mathcal R) \geq 2 - h_2(1/4)$.
\end{lemma}

\begin{proof}
The component $B$ supports one bit per channel use over the noiseless source-destination link. Independently, the source transmits a sub-message through $A \to Y_r$ at any rate below $1 - h_2(1/4)$. The relay decodes this sub-message from one coding block and forwards it through $X_r \to Z$ in the following block. By pipelining over sufficiently many blocks, the initialization and termination overhead vanishes. Thus, the total rate can be arbitrarily close to $2 - h_2(1/4)$.
% We may treat the two information bits $(A, B)$ separately. Note that the channel $\mathcal R$ transmit the information bit $B$ noiselessly to the receiver (up to a permutation in the coordinate). Therefore, we can always use exactly 1 bit to transmit $B$. On the other hand, the channel $\mathcal R$ can be viewed as sending $A$ to $Y_r$ through a binary symmetric channel $\mathrm{BSC}(1/4)$, then decode-forward and send the relay source $X_r$ to $Z$ through another binary symmetric channel $\mathrm{BSC}(1/4)$. Therefore, the rate for transmitting $A$ is
% \begin{align*}
%     \min\{C_{A\to Y_r}, C_{X_r \to Z}\} = \min\{1- h_2(1/4), 1-h_2(1/4)\} = 1-h_2(1/4)
% \end{align*}
% as the channel capacity of $\mathrm{BSC}(\delta)$ is $1 - h_2(\delta)$. As they are conducted in parallel, the overall rate is $1 + 1 - h_2(1/4) = 2 - h_2(1/4)$.
\end{proof}

\begin{remark}
    One may show that the capacity of $\mathcal R$ is exactly $C(\mathcal R) = 2 - h_2(1/4)$ by considering the cutset upper bound as in \cite{el11}. However, a lower bound for $C(\mathcal R)$ suffices to serve as a counterexample.
\end{remark}
\subsection{Capacity from \cite{pon26}}

We upper bounds the quantities $R_{\mathrm{DF}}, R_{\mathrm{U-CF}}, R_{\mathrm{C-CF}}$, respectively. Recall that 
\begin{align*}
    R_{\mathrm{DF}} & = \max_{p(x_s, x_r)} \min\{I(X_s; Y_r | X_r), I(X_s, X_r; Y_d)\} \\
    & = \max_{p(a,b,x_r)} \min\{I(A,B; Y_r| X_r), I(A, B, X_r; B, Z)\}.
\end{align*}
For any input distribution $p(a,b,x_r)$, we have
\begin{align*}
    I(A, B; Y_r |X_r) & = H(Y_r |X_r) - H(Y_r | A, B, X_r) \\
    & = H(Y_r | X_r) - H(Y_r | A) \\
    & \leq H(Y_r) - H(N_1) \\
    & \leq 1 - h_2(1/4).
\end{align*}
On the other hand,
\begin{align*}
    I(A, B, X_r; B, Z) & = H(B, Z) - H(B, Z| A,B, X_r) \\
    & \leq H(B) + H(Z) - \left( \underbrace{H(B|A,B, X_r)}_{= 0} +  H(Z|A,B,X_r)\right) \\
    & \leq 2 - H(Z|X_r) \\
    & = 2 - H(N_2) \\
    & = 2 - h_2(1/4).
\end{align*}
Since this is true for all input distribution $p(a,b,x)$,
\begin{align*}
    R_{\mathrm{DF}} \leq \max_{p(a,b,x)} \min\{1-h_2(1/4), 2-h_2(1/4)\} = 1 -h_2(1/4).
\end{align*}

Recall the conditions
\begin{align}
    I(\hat{Y}_r; Y_r | X_r) & < I(X_r; Y_d) = I(X_r; B, Z) \label{eq: ind-feasible condition} \\
    I(\hat{Y}_r; Y_r | X_r) & < I(\hat{Y}_r; X_r; Y_d) = I(\hat{Y}_r, X_r; B, Z) \label{eq: cor-feasible condition}
\end{align}
that appear in the definition of $\mathcal P_{\mathrm{ind-feasible}}$ and $\mathcal P_{\mathrm{cor-feasible}}$. Note that we always have
\begin{align*}
    I(\hat{Y}_r, X_r; B,Z) = I(X_r; B, Z) + I(\hat{Y}_r; B, Z | X_r) \geq I(X_r; B, Z),
\end{align*}
so for any fixed joint distribution, condition (\ref{eq: ind-feasible condition}) always implies condition (\ref{eq: cor-feasible condition}). Also, note that the objective function in the definition of $R_{\mathrm{U-CF}}$ and $R_{\mathrm{C-CF}}$ is 
\begin{align*}
    I(X_s; \hat{Y}_r, Y_d | X_r) = I(A, B; \hat{Y}_r, B, Z | X_r),
\end{align*}
we relax the maximization formulation to
\begin{align}
    \max \quad & I(A, B; \hat{Y}_r, B, Z | X_r) \label{eq: optimization problem relaxation} \\ 
    \text{subject to} \quad & I(\hat{Y}_r ; Y_r | X_r) \leq I(\hat{Y}_r, X_r; B, Z) \label{eq: optimization constraint} \\
    & p(a,b,x_r, \hat{y}_r, y_r, z) = p_1(a,b, x_r) \cdot W_{1/4}(y_r|a) \cdot W_{1/4}(z|x_r) \cdot p_2(\hat{y}_r| y_r, x_r), \label{eq: factorization form}
\end{align}
for some distribution $p_1, p_2$, and $W_{1/4}$ denote the stochastic kernel that corresponds to the binary symmetric channel $\mathrm{BSC}(1/4)$. Since we have considered a larger feasible set, if $I^*$ is the optimal value to the optimization problem (\ref{eq: optimization problem relaxation}), then $R_{\mathrm{U-CF}}, R_{\mathrm{C-CF}} \leq I^*$.

For joint distribution $p(a,b,x_r,\hat{y}_r, y_r, z)$ that satisfy the factorization in (\ref{eq: factorization form}), we have the following conditional independence structures:
\begin{align}
    Z \perp (A, B, Y_r, \hat{Y}_r) | X_r \label{eq: markov 1} \\
    \hat{Y}_r \perp (A,B) |(Y_r , X_r) \label{eq: markov 2}.
\end{align}
So we can rewrite the objective function (\ref{eq: optimization problem relaxation}) as 
\begin{align}
    I(A, B; \hat{Y}_r, B, Z| X_r) &= I(A, B; \hat{Y}_r, B | X_r) + \underbrace{I(A, B;Z|\hat{Y}_r, B, X_r)}_{=0 \text{ by (\ref{eq: markov 1})}} \nonumber \\
    & = I(A, B; B | X_r) + I(A, B; \hat{Y}_r | X_r, B) \nonumber \\
    & = H(B|X_r) - \underbrace{H(B|A,B,X_r)}_{=0} + I(A; \hat{Y}_r | X_r, B). \label{eq: rewrite objective function}
\end{align}

Next, the left-hand side of (\ref{eq: optimization constraint}) can be rewritten as 
\begin{align*}
    I(\hat{Y}_r; Y_r | X_r) & =I(\hat{Y}_r ; Y_r | X_r) + \underbrace{I(\hat{Y}_r; B|Y_r , X_r)}_{=0 \text{ by (\ref{eq: markov 2})}} \\
    & = I(\hat{Y}_r; B, Y_r | X_r) \\
    & = I(\hat{Y}_r ; B | X_r) + I(\hat{Y}_r ; Y_r | B, X_r).
\end{align*}
And right-hand side of (\ref{eq: optimization constraint}) can also be rewritten as 
\begin{align*}
    I(\hat{Y}_r, X_r; B, Z) & = I(X_r; B, Z) + I(\hat{Y}_r; B, Z | X_r) \\
    & = I(X_r; B, Z) + I(\hat{Y}_r; B | X_r) + \underbrace{I(\hat{Y}_r; Z|X_r, B)}_{=0 \text{ by (\ref{eq: markov 1})}} \\
    & = I(X_r; B, Z) + I(\hat{Y}_r; B | X_r).
\end{align*}
Subtracting $I(\hat{Y}_r; B | X_r)$ from both sides of (\ref{eq: optimization constraint}) simplifies the constraint to the form
\begin{align}
    I(\hat{Y}_r; Y_r | B, X_r) \leq I(X_r; B, Z). \label{eq: new constraint}
\end{align}

Next, we state the post-processing strong data-processing inequality (SDPI).
\begin{theorem}[Post-SDPI \cite{pol25}] \label{thm: SDPI}
    Given a conditional distribution $P_{Y|X}$, the input-free contraction coefficients is defined as 
    \begin{align*}
        \eta_{\KL}^{(p)}(P_{Y|X}) = \sup_{P_X, P_{U|Y}: I(U;Y) > 0} \left\{ \frac{I(U;X)}{I(U;Y)} : X \to Y \to U \text{ is a Markov chain}\right\}.
    \end{align*}
    For the binary symmetric channel $\mathrm{BSC}(\delta)$, we have
    \begin{align*}
        \eta_{\KL}^{(p)}(\mathrm{BSC}(\delta)) = (1 - 2\delta)^2.
    \end{align*}
\end{theorem}

In particular, conditioned on any pair of $(b,x_r)$ and marginalizing (\ref{eq: factorization form}) over $z$ gives
\begin{align*}
    p(a, y_r, \hat{y}_r| B = b, X_r = x_r) = p_1(a|b, x_r) \cdot W_{1/4}(y_r|a) \cdot p_2(\hat{y}_r | y_r, x_r),
\end{align*}
so $A \to Y_r \to \hat{Y}_r$ forms a Markov chain with $P_{Y_r|A} = \mathrm{BSC}(1/4)$. Therefore, applying Theorem \ref{thm: SDPI} yields
\begin{align*}
    I(A; \hat{Y}_r | B = b , X_r =  x_r) \leq \frac{1}{4}I(Y_r; \hat{Y}_r | B = b, X_r = x_r).
\end{align*}
Summing and averaging over all $(b,x_r)$ with positive probability yields
\begin{align*}
    I(A; \hat{Y}_r | B, X_r) \leq \frac{1}{4}I(Y_r; \hat{Y}_r | B, X_r). \label{eq: contraction from SDPI}
\end{align*}
Therefore, the simplified constraint (\ref{eq: new constraint}) and using Theorem \ref{thm: SDPI}, the expression in (\ref{eq: rewrite objective function}) can be further simplified as

\begin{align*}
    I(A, B;\hat{Y}_r, B, Z | X_r) & = H(B|X_r) + I(A;\hat{Y}_r | X_r, B) \\
    & \leq H(B|X_r) + \frac{1}{4}I(Y_r; \hat{Y}_r | B, X_r) \\
    & \leq H(B|X_r) + \frac{1}{4}I(X_r; B, Z) \\
    & = H(B) - I(B; X_r) + \frac{1}{4}\left[ I(X_r; B) + I(X_r; Z|B)\right] \\
    & = H(B) - \underbrace{\frac{3}{4} I(B; X_r)}_{\geq 0} + \frac{1}{4}\left[H(Z|B) - H(Z|X_r, B) \right] \\
    & \leq H(B) + \frac{1}{4}\left(H(Z|B) - H(Z|X_r)\right) \\
    & \leq 1 + \frac{1}{4}(1 - h_2(1/4)) \\
    & = \frac{5}{4} - \frac{h_2(1/4)}{4}.
\end{align*}

Therefore, according to Theorem \ref{thm: main thm},
\begin{align*}
    C_{\mathrm{prop}} & = \max\{R_{\mathrm{DF}}, R_{\mathrm{U-CF}}, R_{\mathrm{C-CF}}\} \\
    & \leq \max\{1 - h_2(1/4), 5/4 - h_2(1/4) / 4\} \\
    & = \max\{0.18872..., 1.04718... \} \\
    & = 1.04718... \\
    & < 1.18872... \\
    & \leq C(\mathcal R),
\end{align*}
which is strictly less than the achievable rate $2 - h_2(1/4) = 1.18872...$ that we evaluated in Section \ref{sec: lower bound}, hence serving as a counterexample to Theorem \ref{thm: main thm}.

\section{Concluding Remarks} \label{sec: concluding remarks}

We demonstrated a counterexample to Theorem 4 in \cite{pon26}, thereby disproving the correctness of the proposed capacity characterization for general relay channels. In our construction, the relay node decodes only a sub-message by jointly processing its observations within each coding block. Consequently, the probabilities of the decoded estimates are determined by the decoding regions of the source code, not merely by the empirical distributions of the sequences used to label those outputs. This observation corroborates our preceding analysis: the probability assignment used in the converse cannot be transferred to the compression codebook solely on the basis of typicality.

\section*{Acknowledgements}
Chun Hei Michael Shiu would like to thank Prof. Chandra Nair for encouraging him to develop the counterexample into a formal manuscript, and Mr. Hadi Kazemi for carefully checking the counterexample, and for discussing possible simplifications of the proofs.

\newpage

\bibliographystyle{IEEEtran}
\bibliography{references}
\end{document}